\documentclass{article}

\usepackage{amsthm} 
\usepackage{amsfonts}
\usepackage{amsmath}
\usepackage{amssymb}
\usepackage{graphicx}
\newtheorem{theorem}{Theorem}
\newtheorem{proposition}[theorem]{Proposition}
\newtheorem{remark}[theorem]{Remark}

\begin{document}

%\markboth{A. Cattabriga, E. Ercolessi, R. Gozzi, E. Meucci}
%{Kraus operators and symmetric groups}

%%%%%%%%%%%%%%%%%%%%% Publisher's Area please ignore %%%%%%%%%%%%%%%
%
%\catchline{}{}{}{}{}
%
%%%%%%%%%%%%%%%%%%%%%%%%%%%%%%%%%%%%%%%%%%%%%%%%%%%%%%%%%%%%%%%%%%%%

\title{Kraus operators and symmetric groups}

\author{Alessia Cattabriga\footnote{A.~Cattabriga has been supported by the "National Group for Algebraic and Geometric Structures, and their Applications" (GNSAGA-INdAM) and University of Bologna, funds for selected research topics.}, Elisa Ercolessi\footnote{E.E. is partially supported by INFN through the project “QUANTUM” and by the project QuantHEP of the QuantERA ERA-NET Co-fund in Quantum Technologies}, Riccardo Gozzi, Erika Meucci }

%\address{Department of Mathematics, University of Bologna, piazza di Porta San Donato 5\\
%Bologna, 40126, Italy\\
%\email{alessia.cattabriga@unibo.it}}

%\author{ELISA ERCOLESSI}

%\address{Department of Physics and Astronomy, University of Bologna, via Irnerio 46 \\and INFN, Sezione di Bologna\\
%Bologna, 40126, Italy\\
%\email{elisa.ercolessi@unibo.it}}

%\author{RICCARDO GOZZI}

%\address{Instituto Superior Tecnico, Universidade de Lisboa, Av. Rovisco Pais 1, 1049-001 \\and Instituto de Telecomunicaçoes, Lisboa, Portugal\\
%\email{ilgozzi@mail.com}}

%\author{ERIKA MEUCCI}

%\address{School of Advanced International Studies, Johns Hopkins University, via Beniamino Andreatta 3\\ 
%Bologna, 40126, Italy\\
%\email{emeucci@jhu.edu}}

\maketitle

%\begin{history}
%\received{(Day Month Year)}
%\revised{(Day Month Year)}
%\end{history}

\begin{abstract}
In the contest of open quantum systems, we study a class of Kraus operators whose definition relies on the defining representation of the symmetric groups. We analyze the induced orbits as well as the limit set and the degenerate cases.
\end{abstract}

%\keywords{Kraus operator; open quantum systems; symmetric group.}

\section{Introduction and preliminaries}

We are interested in studying open quantum systems, that is systems that are free to interact with the environment or with other systems. The study of open systems is useful in fields such as quantum optics, quantum measurement theory, quantum statistical mechanics and quantum cosmology. Moreover, the study of composite systems is at the heart of quantum computation and quantum information, where, for examples, concepts like entanglement can have applications in devising algorithms and protocols, such as quantum teleportation, that do not have a classical analogue.

In elementary quantum mechanics, the state of a closed quantum system is represented by a ray \cite{EMM} in a separable Hilbert space $\mathbf{v} \in {\cal H}$, i.e. by an equivalence class of vectors $[ \mathbf{v} ]$, $\mathbf{v} \in {\cal H}$, with respect to the relation: $\mathbf{v} \sim \lambda \mathbf{v} $ with $\lambda \in \mathbb C - \left\{ 0\right\}$. Such equivalence class can be represented via the density matrix $\rho_{\mathbf{v}} \equiv \mathbf{v} \mathbf{v}^\dagger / \| \mathbf{v}\|^2$ (where $\mathbf{v}^\dagger$ is in the dual space ${\cal H}^*\simeq {\cal H} $), which is a rank-one projector or, more precisely,  a bounded, self-adjoint, positive definite, unit-trace operator such that:
\begin{equation} \label{pure}  
\rho_{\mathbf{v}} ^2=\rho_{\mathbf{v}} . 
\end{equation}
A system whose density matrix satisfies the latter condition is said to be {\it pure}. In general, as we will explain shortly below, it is necessary to consider a more general kind of density matrices that are constructed out of a statistical mixture $\{ \rho_j,p_j\}_{j=1}^N$, where the $\rho_j$'s are pure density matrices and the $p_j$'s are probabilities, therefore satisfying $0\leq p_j\leq 1$ and $\sum_{j=1}^N p_j =1$. Such a density matrix (called {\it mixed}) is obtained by setting 
\begin{equation}\label{mixed}
\rho = \sum_{j=1}^N p_j \rho_j
\end{equation}
and is a bounded, self-adjoint, positive definite, unit-trace operator, with now $\rho^2\neq \rho $.\\
In the following we will be interested in the case in which the Hilbert space is finite dimensional with $\dim({\cal H})=n$, a fact we will assume from now on.\\
When considering an open quantum system, i.e. a (sub)system $A$ in interaction with an environment $B$, the Hilbert space representing the total system is given by $\mathcal H_A\otimes\mathcal H_B$, where the general element of $\mathcal H_A\otimes\mathcal H_B$ will be $\psi_{AB}=\sum_{i\in I, j\in J}a_{ij}\mathbf e^A_i\otimes \mathbf e^B_j$ where $\{\mathbf e^A_i\mid i\in I\}$ and  $\{\mathbf e^B_j\mid j\in J\}$ are orthonormal basis of $\mathcal H_A$ and  $\mathcal H_B$ respectively and $\sum_{i\in I, j\in J}|a_{ij}|^2=1$. The density matrix representing the quantum state of the subsystem $A$ is obtained by taking the partial trace over the environment $B$: $\rho_A = \textup{Tr}_B[ \rho_{\psi_{AB}}]$. A very well known theorem \cite{PRE} states that the density matrix $\rho_A$ is pure if and only if $\psi_{AB}$ is of the form $\mathbf v_A\otimes \mathbf v_B$ with $\mathbf v_A\in\mathcal H_A$ and $\mathbf v_B\in\mathcal H_B$, i.e. it is a {\it separable} state.  In all other cases, i.e. when the state $\psi_{AB}$ is entangled, $\rho_A$ will represent a mixed state. \\
The space of quantum states can be endowed of interesting geometrical structures. The set of pure states is a complex projective space and indeed a Kahler manifold \cite{EMM}, that can be embedded in the the Lie algebra of self-adjoint matrices as the (co)-adjoint orbit of the unitary group $U(N)$ of rank-one density matrices. The latter description generalizes also to rank-k mixed states \cite{CCM,CES}, so that the full space of states can be seen as the union of orbits of the unitary groups, each of them been a complex manifold endowed by a metric, a symplectic form and a compatible complex structure.\\
From a more algebraic point of view, let us notice that the space of pure states can be seen as the extremal points of the positive cone in the algebra of self-adjoint operators, generated by positive definite and unit-trace matrices. 

The time evolution of a closed quantum system is determined by the Schroedinger equation \cite{PRE} or, when the Hamiltonian operator $\mathbb H$ is time-independent, by the unitary operator $U(t) = \exp[-i \mathbb H t/\hbar]$ via:
\begin{equation} \label{sch}
\rho(t) = U(t) \rho(t=0) U(t)^\dagger. 
\end{equation}
This is the evolution also of the density matrix of the total system $A\cup B$, which can be seen as isolated, whereas the evolution of the subsystem $A$ which is obtained by taking the partial trace: $\rho_A(t) = \textup{Tr}_B[ \rho_{\psi_{AB}}(t)] $. Let us notice that,
contrary to what happens for a closed system, open quantum dynamics may change the spectrum as well as the rank of the density matrix. \\
The dynamics of open quantum systems is generated by the so-called Gorini-Kossakowski-Sudarshan- Linbland (GKLS) equation (see ref. \cite{CP} for a, also historical, review):
\begin{equation}\label{GKLS}
\mathbf L(\rho)=-i[\mathbb H,\rho]-\frac12 \sum_{j=1}^N \left\{V_j^{\dagger} V_j, \rho\right\}+\sum_j  V_j\rho V_j^{\dagger},
\end{equation}
where $\mathbb H$ is the Hermitian Hamiltonian operator and the $V_j$'s ($N=1,2,\cdots ,n^2-1$) are arbitrary (bounded) operators\footnote{As usual we denote with $V^{\dagger}$ the adjoint of $V$, with $[\cdot,\cdot]$ the commutator and with  $\{\cdot,\cdot\}$  the  anticommutator.}. \\
Eq. (\ref{GKLS}) has a clear geometrical interpretation. Indeed, it can be shown \cite{CCILM} that the dynamical evolution described by this equation defines a vector field $\Gamma= X+Y+Z,$ that can be decomposed into three vector fields related with the three addends of the GKLS equation. More precisely, $X$ is a Hamiltonian vector field whose flow preserves the spectrum of $\rho$ (hence moving on a given co-adjoint orbit), $Y$ is a gradient-like vector field whose flow changes the spectrum but preserves the rank while $Z$  is a vector field corresponding to a flow that changes the rank of $\rho$.      

In this paper we take  a different, but equivalent perspective \cite{RH}, according to which the dynamic of an open quantum system with density matrix $ \rho_A(t_0)\equiv\rho(t_0)$ is described by means of the so called Universal Dynamical Maps (UDM), that is a trace-preserving linear completely positive\footnote{Here for completely positive we mean that the operator  $K_{[t_0,t_1]}\otimes 1_B$  is positive for any possible extension of $\mathcal H_A$ to $\mathcal H_A\otimes\mathcal H_B$. We refer the interested reader to \cite{CP} for a detailed discussion of why completely positivity and not simply positivity is required.} definite map defined as
\begin{equation}\label{UDM}
{\cal E}_{K([t_0,t_1])}:\; \rho(t_0)\to \rho(t_1) = \sum_{\alpha} K_{\alpha}(t_1,t_0)\rho(t_0)K_{\alpha}(t_1,t_0)^{\dagger}
\end{equation}
given an initial configuration of the system at time $t=t_0$ encoded by the density matrix $\rho(t_0)$. The operators  $K_{\alpha}(t_1,t_0)$, for $\alpha\in\mathcal A$ are called \textit{Kraus operators}: they do not depend on the initial condition $\rho(t_0)$, but, as the indices should suggest, just on the time interval $[t_0,t_1]$. We will call Kraus map a linear combination of Kraus operators. Moreover, to ensure $\textup{Tr}[\rho(t_1)]=1$, the Kraus operators must satisfy the following condition
\begin{equation}\label{Krauss}
\sum_{\alpha\in\mathcal A} K_{\alpha}(t_1,t_0)K_{\alpha}(t_1,t_0)^{\dagger}=\mathbf 1.
\end{equation}
Whenever the (super)-operators (\ref{UDM}) satisfy also:
\begin{equation}\label{semigroup}
{\cal E}_{K([s,t])}\circ {\cal E}_{K[(t,0)]}  = {\cal E}_{K([s,0])} , \, \forall s \geq t \geq 0, 
\end{equation}
eq. (\ref{UDM}) defines a one-parameter semigroup of completely positive maps, of which the operator $\mathbf L$ of eq. (\ref{GKLS}) is the generator. In this case the dynamics is called quantum Markovian \cite{RH}. Le us remark that the decomposition of such a map into Kraus operators is not unique, a question that will be considered in the following. 

In this paper, we are interested in characterizing the non-unitary part of the dynamic and we set $\mathbb H=0$. 
Notice that the GKLS equation is invariant under any unitary transformation, since $\rho(t) \to U \rho(t) U^\dagger$, $V_j \to U V_j U^\dagger$, for all $j$. Also, we can always find a unitary transformation such that, at the initial time, we can write the density matrix in the diagonal form $\rho(t=0)=\textup{diag}(\lambda_1,\ldots,\lambda_n)$, with $\lambda_i\geq 0$ and $\lambda_1+\cdots+\lambda_n=1$. Thus we can consider a GKLS equation of the form 
\begin{equation}\label{GKLS1}
\mathbf L(\rho)= -\frac{1}{2} \sum_{j=1}^N\left\{V_j^{\dagger}V_j, \rho\right\}+\sum_{j=1}^N V_j\rho V_j^{\dagger},
\end{equation}
with  $\rho$ diagonal.  This situation encompasses a series of interesting cases in physics, such as (when N=1) the so-called Quantum Poisson and Gaussian Semigroups \cite{CCILM}. 

In order to describe explicitly this type of dynamics, we take an algebraic  approach and consider Kraus operators associated to elements of the symmetric group $\Sigma_n$ via the defining representation. In Section \ref{simmetrico}, we  recall the basic notions of this representation  and associate a  Kraus map  to each element  of $\mathbb C[\Sigma_n]$, the group algebra of $\Sigma_n$. We give conditions on the elements of $\mathbb C[\Sigma_n]$ giving rise to  Kraus maps with admissible action and reduce the study of the orbits of the dynamic to those associated to cyclic subgroups of $\Sigma_n$. In Section \ref{azione}, we compute explicitly the  orbits in the cyclic case  as well as the limit set of the dynamics. We  will also describe what happens in the degenerate cases, that is the cases in which the initial density matrix is not generic (i.e., the cardinality of the spectrum of $\rho$ is less than the order of the matrix).

\section{Defining representation of $\Sigma_n$ and associated Kraus maps}
\label{simmetrico}
In this section, after recalling  some classical notion on the defining representation of $\Sigma_n$ (see \cite{S}), we describe how to associate a Kraus map to elements of $\mathbb C[\Sigma_n]$, the group algebra of $\Sigma_n$.\\

Let $\Sigma_n$ be the symmetric group on $n$ letters. The $n$-dimensional \textit{defining} representation $\chi:\Sigma_n\to GL_n(\mathbb C)$ of $\Sigma_n$ is given by %
\begin{equation}
\chi(\sigma)=R_{\sigma}\qquad\qquad \textup{with\ } \left(R_{\sigma}\right)_{ij}=\left\{\begin{array}{l}1\quad\textup{if } \sigma(j)=i\\
                                                                              0 \quad\textup{otherwise}
                                                                             \end{array}\right. . 
\end{equation}\label{rep}
The defining representation is unitary (i.e., the image of $\chi$ is contained in $U(n)$) and  reducible. Indeed, the 1-dimensional subspace $W$ spanned by $\mathbf e_1+\mathbf e_2+\cdots +\mathbf e_n$, with $\mathbf e_i$ the $i$-th vector of the canonical basis of $\mathbb C^n$, is invariant under the action of $\chi(\Sigma_n)$ and $\chi$ restricts to the trivial action on $GL(W)$. Moreover, $\chi$ is completely reducible: indeed, if  $W^{\perp}$ denotes the orthogonal complement with respect to the standard hermitian product on $\mathbb C^n$, also  $W^{\perp}$ is  invariant under the action of  $\chi(\Sigma_n)$. It is also possible to prove that the $(n-1)$-dimensional representation of $\Sigma_n$ induced by $\chi$ into $GL(W^{\perp})$ is irreducible.\\
If we identify $\mathbb C^n$ with the vector space $D_n(\mathbb C)$ of diagonal matrices with complex entries, then  $\Sigma_n$ acts on $D_n(\mathbb C)$ as 
$$\sigma \cdot \textup{diag}(\lambda_1,\ldots,\lambda_n)= R_{\sigma}  \textup{diag}(\lambda_1,\ldots\lambda_n) R_{\sigma}^{-1} =\textup{diag}(\lambda_{\sigma(1)},\ldots,\lambda_{\sigma(n)}).$$
Since $W$ is an invariant subspace with respect to $\chi$, then the trace of a matrix is invariant under this action.

Let $\mathbb C[\Sigma_n]$ denotes the group algebra of $\Sigma_n$. The defining representation of $\Sigma_n$ naturally induces a representation of   $\mathbb C[\Sigma_n]$ into $M_n(\mathbb C)$, that we still denote with $\chi$, given by 
\begin{equation}\label{galg}
\chi\left(\sum_{\sigma\in\Sigma_n}c_{\sigma}\sigma \right)\mathbf v=\sum_{\sigma\in\Sigma_n}c_{\sigma}\left(R_{\sigma}\mathbf v\right),
\end{equation}
with $c_{\sigma}\in\mathbb C$, for all $\sigma$,  and $\mathbf v\in\mathbb C^n$. 

We want now to introduce a time dependence: for each $\sigma\in\Sigma_n$ we choose $\mathbb C$-valued smooth functions $[0,+\infty)\ni t \mapsto c_{\sigma}(t)\in \mathbb C$ and   consider the map  
\begin{equation}\label{time}
[0,+\infty)\ni t\mapsto \sum_{\sigma\in\Sigma_n}c_{\sigma}(t)\sigma\in \mathbb C[\Sigma_n].
\end{equation}
If we look at the operator obtained trough $\chi$, we have that:
$$\sum_{\sigma\in\Sigma_n}\left(c_{\sigma}(t)R_{\sigma}\right) \left(c_{\sigma}(t)R_{\sigma}\right)^{\dagger}=\sum_{\sigma\in\Sigma_n}c_{\sigma}(t)
\overline{c}_{\sigma}(t)R_{\sigma}R_{\sigma}^{\dagger}=\left(\sum_{\sigma\in\Sigma_n}c_{\sigma}(t)\overline{c}_{\sigma}(t)\right) \textup{Id}_n,$$
where $\overline{c}_{\sigma}(t)$ denotes the complex conjugate function. So if 
\begin{equation}\label{condition}
\sum_{\sigma\in\Sigma_n}c_{\sigma}(t)\overline{c}_{\sigma}(t)=\mathbf{1},
 \end{equation}
where $\mathbf{1}$ denotes the constant function equal to 1,  the element $\sum_{\sigma\in\Sigma_n}c_{\sigma}(t)\sigma$ acts on $D_n(\mathbb C)$,  via $\chi$, as a Kraus operator. 

Even if this condition on the coefficients is satisfied, a generic element of $\sum_{\sigma\in\Sigma_n}c_{\sigma}(t)\sigma\in\mathbb C[\Sigma_n]$ might not yield a suitable one-parameter semigroup unless: i) it is completely positive; ii) it satisfies the time condition (\ref{semigroup}). \\

As for the second condition, let us notice that,  if we consider 
$$\chi\left(\sum_{\sigma\in\Sigma_n}c_{\sigma}(t)\sigma \right)=\sum_{\sigma\in\Sigma_n}c_{\sigma}(t)R_{\sigma},$$ 
such that $\mathcal N=\{\sigma\in\Sigma_n\mid c_{\sigma}(t)\ne \mathbf 0\}$ is a subgroup of $\Sigma_n$, it is possible to choose opportunely the coefficients $c_{\sigma}(t)$ so that the operators satisfy the time conditions. Indeed, given a subgroup $S$ of $\Sigma_n$, we  associate to it the operator  
\begin{equation}
K_{S}=g(t) \textup{Id}_n + f(t) \sum_{\sigma\in S} R_{\sigma}
\end{equation}
giving  the evolution function
\begin{equation}
F^S_{(t,0)}(\rho(0))=\rho(t) = g^2(t) \rho(0) + f^2(t) \sum_{\sigma\in S} R_{\sigma} \rho(0) R_{\sigma}^{-1},
\end{equation}
where
\begin{equation}
g(t) = \sqrt{\frac{1}{|S|}\left(1+ (|S |-1) e^{-t}\right)} \mbox{ and } f(t) = \sqrt{\frac{1}{|S|}(1-e^{-t})},
\end{equation}
with $|S|$ the order of $S$. It is immediate to check that $K_{S}$ satisfies (\ref{condition}). \\

Given $S$ and $T$ subgroups of $\Sigma_n$, we say that $K_S$ and $K_T$ are \emph{equivalent}, and write $K_S\cong K_T$, if they determine the same evolution function that is 
$$F^S_{(t,0)}(\rho(0))=F^T_{(t,0)}(\rho(0))$$
for each $t\in[0,+\infty)$ and each $\rho(0)\in D_n(\mathbb C)$. \\
Since our aim is to study all the possible evolution functions, we want to look at Kraus maps associated to subgroups of $\Sigma_n$ up to equivalence. On this regard, we have the following result. 

\begin{proposition}  \label{equivKop}
Two Kraus maps $K_S$ and $K_T$, associated to subgroups $S,T\subseteq\Sigma_n$, are equivalent if and only if the partition of $\{1,2,\ldots,n\}$  associated to the orbits of the action of $S$ and $T$ onto $\{1,2,\ldots,n\}$ is the same. 
\end{proposition}
\begin{proof}
Notice that $K_S$ and $K_T$  are equivalent if and only if $|S|=|T|$ and 
$$
\sum_{\sigma \in T} R_{\sigma}  \rho(0) R_{\sigma}^{-1} = \sum_{\sigma' \in S} R_{\sigma'}  \rho(0) R_{\sigma'}^{-1}
$$
Therefore, there is a bijective map $T \rightarrow S$ that sends $\sigma \in T$ in $\sigma' \in S$.
Since the $R$'s matrices are permutation matrices, this happens if and only if the orbits of the action of $S$ and $T$ onto $\{1,2,\ldots,n\}$ are the same.
\end{proof}

The previous proposition allows us to reduce to the case in which $S$ is a cyclic subgroup. Indeed, it is enough to select one subgroup $S$ for each partition of $\{1,2,\ldots,n\}$ and we can always choose a cyclic subgroup: given a partition $\{p_1,\dots,p_k\}$ of $\{1,2,\ldots,n\}$  we can take the cyclic subgroup ge\-ne\-rated by $\sigma=c_1\cdots c_k$, where $c_i$ is any cycle permuting all the elements  of $p_i$, for $i=1,\ldots,k$. In other words, given a diagonal matrix $\rho$, it is always possible to find a cyclic subgroup that permutes the elements on the matrix according to the evolution. \\

Let us consider now the following trivial but important:
\begin{remark} \label{conjsub}
If $S$ and $T$ are conjugated subgroups, then the evolution function $F^{T}_{(t,0)}$ can be deduced from the evolution function $F^{S}_{(t,0)}$ because 
\begin{itemize}
\item $|S|=|T|$ and hence $g_T(t) = g_S(t)$ and $f_T(t) = f_S(t)$;
\item $$\sum_{\sigma \in T} R_{\sigma}  \rho(0) R_{\sigma}^{-1} \, = \sum_{\sigma' = \tau \sigma \tau^{-1} \in S} R_{\sigma'}  \rho(0) R_{\sigma'}^{-1},$$ where $\tau \in T$ and $R_{\sigma'} = R_{\tau} R_{\sigma} R_{\tau}^{-1}$.
\end{itemize} 
\end{remark}
Therefore, by Proposition \ref{equivKop}  and Remark \ref{conjsub}, in order to understand the dynamical evolution, it is enough to consider actions of cyclic subgroups of $\Sigma_n$ onto $\mathbb C^n$ up to conjugation. It is important to recall that the number of conjugacy classes of elements in $\Sigma_n$, that corresponds to the number of cyclic subgroups of $\Sigma_n$ up to conjugacy,  depends  on the number of partitions of $n$ as follows. First we recall that a partition $\mu$ of $n$ is a vector $(\mu_1,\ldots,\mu_r)$, whose entries are positive integers and satisfy $\mu_1+\cdots\mu_r=n$ and $\mu_i\geq \mu_{i+1}$. Given an element $\sigma\in\Sigma_n$, let $\sigma=c_{1}\cdots c_{r}$ be its decomposition into disjoint cycles and, up to renumbering the cycles,  suppose that $|c_i|\geq |c_{i+1}|$, where $|c_i|$ denotes the length of the $i$-th cycle. We can associate to $\sigma$  a partition $\mu_{\sigma}$ of $n$ given by $(|c_1|,\ldots, |c_r|)$. The following facts hold:
\begin{itemize}
 \item[a)] given two element $\sigma_1,\sigma_2\in\Sigma_n$, they are conjugated if and only if $\mu_{\sigma_1}=\mu_{\sigma_2}$;
 \item[b)] the map $[\sigma]\to\mu_{\sigma}$ is a one to one correspondence on the level of conjugacy classes of elements in $\Sigma_{n}$.
\end{itemize}

\section{Action of Kraus maps associated to subgroups  of $\mathbf{\Sigma_n}$}
\label{azione}

Given $S=\langle \sigma\rangle$ be a cyclic subgroup of $\Sigma_n$, we have 
\begin{equation}
K_S=K_{\sigma}=g(t) \textup{Id}_n + f(t) \sum_{i=1}^{|\sigma|-1} R_{\sigma}^{i}
\end{equation}
and 
\begin{equation}
\rho(t) = g^2(t) \rho(0) + f^2(t) \sum_{i=1}^{|\sigma|-1} R_{\sigma}^{i}  \rho(0) R_{\sigma}^{-i}.
\end{equation}
We want to compute an explicit analytic expression for $\rho(t)$.  

Suppose that $\sigma=c_1\cdots c_r$ is the decomposition into disjoint  cycles, including cycles of length one, and with  $|c_i|\geq |c_{i+1}|$. Let  $|c_i|=\mu_i$, for $i=1,\ldots,r$ and set $\mu_0=1$. Clearly $\mu_1+\mu_2+\cdots \mu_r=n$. Since we work up to conjugacy, we can assume that the permutation has the following form
$$\sigma=\left(1\quad 2 \cdots \mu_1\right)\left(\mu_1+1\quad \mu_1+2 \cdots \mu_1+\mu_2\right)\cdots \left(\sum_{j=1}^{r-1}\mu_i\quad 1+\sum_{j=1}^{r-1}\mu_i\cdots n\right).$$
So the $i$-th cycle is 
$$c_i=\left(\sum_{j=0}^{i-1}\mu_j \quad  1+\sum_{j=0}^{i-1}\mu_j \quad \cdots \quad \sum_{j=1}^{i}\mu_j\right).$$
Note that $|\sigma|=\textup{LCM}\{\mu_1,\ldots, \mu_r\}$ , where $\textup{LCM}$ stands for  the least common multiple.

A straightforward computation shows that if $\rho(0)=\textup{diag}(\lambda_1,\ldots, \lambda_n)$ and we set 
\begin{equation}\label{B_i}
B_i=\left(\frac{1}{\mu_i}\sum_{h\in c_i}\lambda_h\right) \textup{Id}_{\mu_i}
\end{equation}
for $i=1,\ldots, r$, then 
\begin{equation}
\label{rho}
 \rho_{\sigma}(t)=\rho(0) e^{-t}+(1-e^{-t})B,
\end{equation}
where $B$ is the block diagonal matrix $B_1\oplus B_2\oplus \cdots \oplus B_r$.\\
We observe that, by equations \eqref{B_i} and \eqref{rho}, the eigenvalues of the matrix $\rho_{\sigma}(t)$ are linear combinations of the eigenvalues of $\rho(0)$ with non negative coefficients and at least one non-zero coefficient.

We can then formulate the following:
\begin{theorem}
The action associated to each cyclic subgroup $< \sigma >$ of $\Sigma_n$ satisfies two properties: 1) it is completely positive and 2) it satisfies the time condition (\ref{semigroup}).
\end{theorem}

\begin{proof}
We consider   
\begin{equation}
\rho(t) = g^2(t) \rho(0) + f^2(t) \sum_{i=1}^{|\sigma|-1} R_{\sigma}^{i}  \rho(0) R_{\sigma}^{-i}.
\end{equation}
Since a map $A \rightarrow BAB^*$ is completely positive \cite{Choi}, and the sum of completely positive operators is completely positive, we can deduce that the action associated to each cyclic subgroup $< \sigma >$ is completely positive.  

Now we prove the time condition, $F_{(s,t)} \circ F_{(t,0)} (\rho(0)) = F_{(s,0)} (\rho(0))$. Noticing that
$$
F_{(s,0)}(\rho(0)) = \rho_{\sigma}(s) = g^2(s) \rho(0) + f^2(s) \sum_{i=1}^{|\sigma|-1} R_{\sigma}^{i}  \rho(0) R_{\sigma}^{-i}, 
$$
we have
\[
\left.
\begin{array}{l}
F_{(s,t)} \circ F_{(t,0)} (\rho(0)) = g^2(s-t) g^2(t) \rho(0) + g^2(s-t) f^2(t) \sum_{i=1}^{| \sigma | -1} R_{\sigma}^{i}  \rho(0) R_{\sigma}^{-i} + \\
+ f^2(s-t) g^2(t) \sum_{i=1}^{|\sigma|-1} R_{\sigma}^{i}  \rho(0) R_{\sigma}^{-i} + f^2(s-t) f^2(t) \sum_{i,j=1}^{|\sigma|-1} R_{\sigma}^{i+j}  \rho(0) R_{\sigma}^{-(i+j)} = \\
= \{ g^2(s-t) g^2(t) + ( | \sigma | -1) f^2(s-t) f^2(t)  \} \rho(0) + \\
+  \sum_{i=1}^{|\sigma|-1} \{ g^2(s-t) f^2(t) +  f^2(s-t) g^2(t) + (| \sigma | -2) f^2(s-t) f^2(t) \} R_{\sigma}^{i}  \rho(0) R_{\sigma}^{-i} = \\
= \frac{1}{| \sigma |^2} \{ [1+ (| \sigma | -1) e^{-s+t}] [1+ (| \sigma | -1) e^{-t}] + ( | \sigma| -1) [1- e^{-s+t}] [1-e^{-t}] \} \rho(0) + \\
+ \frac{1}{| \sigma |^2} \sum_{i=1}^{| \sigma | -1} \{ [1+ (| \sigma | -1) e^{-s+t}]  [1-e^{-t}] + [1-e^{-s+t}] [1+ (| \sigma | -1) e^{-t}]  + \\ 
+ (| \sigma | -2) [1- e^{-s+t}] [1-e^{-t}]  \} R_{\sigma}^{i}  \rho(0) R_{\sigma}^{-i} =  \\
= \frac{1}{| \sigma |^2} \{ 1+ ( | \sigma| -1) e^{-s+t} + ( | \sigma| -1) e^{-t} + ( | \sigma| -1)^2 e^{-s} + | \sigma| -1 - ( | \sigma| -1) e^{-s+t} + \\
- ( | \sigma| -1) e^{-t} + ( | \sigma| -1) e^{-s} \} \rho(0) + \frac{1}{| \sigma |^2} \sum_{i=1}^{| \sigma | -1} \{ 1 - e^{-t}+ (| \sigma | -1) e^{-s+t} -  (| \sigma | -1) e^{-s} + \\
+1 - e^{-s+t} +  (| \sigma | -1) e^{-t} -  (| \sigma | -1) e^{-s} +  | \sigma | -2 - (| \sigma | -2) e^{-s+t} - (| \sigma | -2) e^{-t} + \\
+ (| \sigma | -2) e^{-s} \} R_{\sigma}^{i}  \rho(0) R_{\sigma}^{-i} =  \\
 = \frac{1}{|\sigma|}[1+ (| \sigma | -1) e^{-s}] \rho(0) +  \frac{1}{|\sigma|}(1-e^{-s}) \sum_{i=1}^{|\sigma|-1} R_{\sigma}^{i}  \rho(0) R_{\sigma}^{-i} = \rho_{\sigma}(s).
\end{array}
\right. 
\]
\end{proof}

Using the explicit description of Formula (\ref{rho}), we can easily deduce a  description of the associated orbit. First of all, notice that 
\begin{equation}
 \label{lim}
\lim_{t\to +\infty }\rho_{\sigma}(t)=B.
\end{equation}
Moreover we have $\rho(0)-\rho_{\sigma}(t)=(1-e^{-t})(\rho(0)-B),$
so it is easy to check that $\rho(0)-\rho(t)$ is a diagonal matrix whose diagonal entries satisfy the system of equations
\begin{equation}
\label{system}
\left\{\begin{array}{l}x_1+x_2+\cdots +x_{\mu_1}=0\\
   x_{\mu_1+1}+x_{\mu_1+2}+\cdots +x_{\mu_1+\mu_2}=0  \\
   \vdots\\
    x_{n-\mu_r}+x_{n-\mu_{r}+1}+\cdots +x_{n}=0.        
\end{array}\right.
\end{equation}

Notice that  if we start with a matrix $\rho(0)$ having $n$ different eigenvalues, that is a \textit{generic} initial condition, the limit point of the orbit of $K_{\sigma}$ with  $\sigma=c_1\cdots c_r$, lies in a closed subspace containing the matrices having at most $r$ distinct eigen\-va\-lues $\nu_1,\ldots, \nu_r$. 

If we start with a matrix $\rho(0)$ having at least one eigenvalue with multiplicity greater then one, all the previous results hold. Nevertheless, we have less freedom in movements: indeed, there exists non-trivial elements of $\Sigma_n$ acting trivially on the submanifold containing it. More precisely, the elements of $\Sigma_n$ that permutes the eigenvalues that are equal act trivially on $\rho(0)$. In order to have a non-trivial action, the different eigenvalues must be shuffled by the permutation. Hence, if, given a partition $\lambda=(\lambda_1,\ldots,\lambda_r)$ of $n$, we denote with $M_{\lambda}$ the submanifold containing matrices having $r$-different eigenvalues with multiplicities $\lambda_1,\ldots,\lambda_r$, the trivial action is carried by the stabilizer  of $M_{\lambda}$  in  $\Sigma_n$. This subgroup can be characterized as that containing $\sigma$ such that $\mu_{\sigma}=(\mu_1,\ldots,\mu_k)$ is a subpartition of $\lambda$, that is there exist indices $1\leq j_1\leq j_2\leq \ldots\leq j_r\leq k$ such that $\lambda_i=\mu_{j_i}+\mu_{j_i+1}+\cdots +\mu_{j_{i+1}-1}+\mu_{j_{i+1}}$, for $i=1,\ldots,r$.

\section{Geometric interpretation and examples}

Let's try to have a more geometric picture. Given $n$ points $P_1,\ldots, P_{n}\in \mathbb C^n$ in general position, we denote the $(n-1)$-simplex having  $P_1,\ldots, P_{n}$ as vertices with $\Delta^{n-1}=\Delta(P_1,\ldots,P_n)$. To each matrix $\rho(0)=\textup{diag}{(\lambda_1,\lambda_2,\ldots,\lambda_n)}$ we can associate a point $\lambda_1 P_1+\lambda_2 P_2+\cdots+\lambda_n P_n$ in $\Delta(P_1,\ldots,P_n)$. Each element $\sigma\in\Sigma_n$ clearly acts on the vertices of $\Delta^{n-1}$ and, by linearity, on the points of the simplex. Given a cycle $c=(i_1\ i_1\ \cdots i_{\mu})$, we denote by $L(c)\subset \mathbb C^n$ the  subspace spanned by  the $n-1$ vectors $P_{i_1}-P_{i_j}$, with $j=2,\ldots,\mu$. Moreover, with the notation $\textup{Bar}(c)$ we indicate the barycenter of the $(\mu-1)$-simplex having vertices $P_{i_1},\ldots, P_{i_{\mu}}$.  Notice that $c$ fixes $\textup{Bar}(c)$. 

Since $K_{\sigma}$ satisfies the range condition, the orbit $\rho_{\sigma}(t)$ gives  a path inside $\Delta^{n-1}$.  In this setting, (\ref{system}) tells us that the orbit is contained in the affine subspace passing trough the point associated to $\rho(0)$ and parallel to the subspace $L(c_1)\oplus L(c_2)\oplus\cdots\oplus L(c_r)$, with $\sigma=c_1\cdots c_r$. Moreover, the limit of the orbit is the intersection point between this  affine subspace and the affine closure of the points  $\textup{Bar}(c_1),\ldots,\textup{Bar}(c_r)$.\\

 \begin{figure}[h]
    \centering\includegraphics[height=60mm]{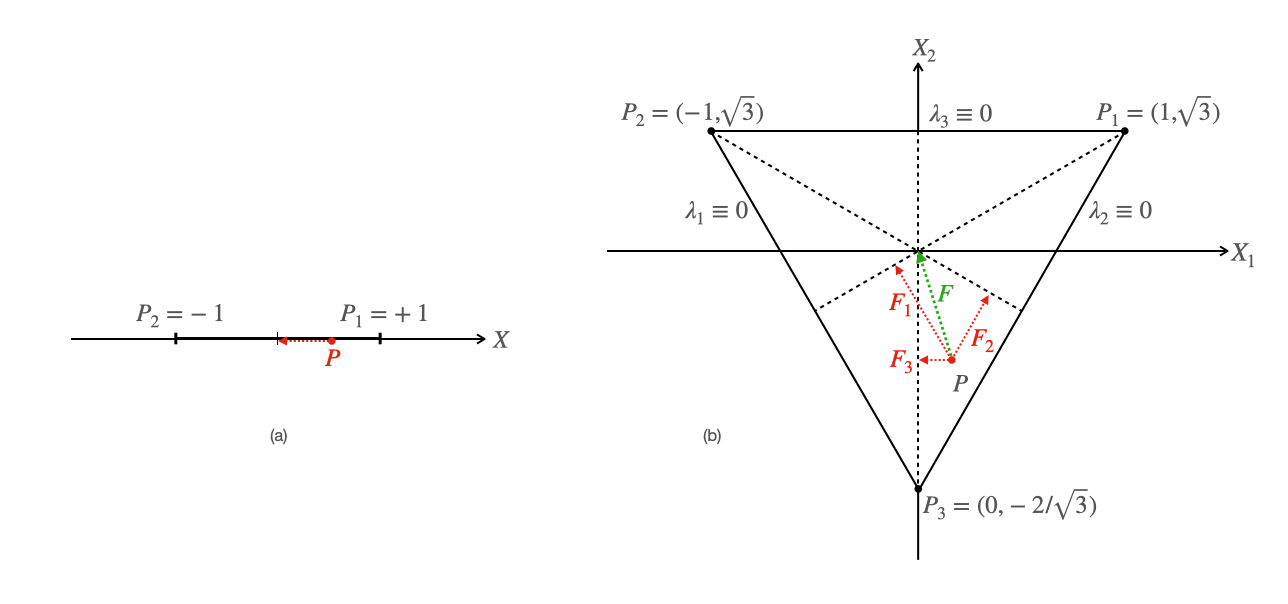}
    \caption{The simplex of diagonal density matrices for (a) $n=1$, i.e. a qubit, and (b) $n=3$, i.e. a qutrit. The arrows show the direction of the time evolution of the operators defined in the text. \label{fig:fig}}
   \end{figure}

As a first example, we can take the case of a qubit, i.e. a Hilbert space of dimension $n=2$, so that a generic diagonal density matrix is of the form $\rho = \textup{diag}(\lambda_1,\lambda_2=1-\lambda_1)$. By setting $X=\lambda_1-\lambda_2 \in [-1,1]$, we can represent the simplex $\Delta^1$ as the segment (convex cone) generated by the two points $P_1, P_2$ with coordinate $X=+1,-1$ respectively, as shown in Fig. 1(a). \\
The unique non-trivial cyclic subgroup is now generated by the permutation $\sigma : (\lambda_1,\lambda_2)\mapsto (\lambda_2,\lambda_1)$. A simple calculation shows that it generates the time evolution:
\begin{eqnarray}
 \rho_{\sigma}(t) & =\textup{diag} ( e^{-t}  \lambda_1 + (1-e^{-t} ) (\lambda_1+\lambda_2)/2 \, , \, e^{-t}  \lambda_2 + (1-e^{-t} ) (\lambda_1+\lambda_2)/2 ) \nonumber \\
 & =\textup{diag} ( e^{-t}  \lambda_1 + (1-e^{-t} ) /2 \, , \, e^{-t}  \lambda_2 + (1-e^{-t} ) /2 ) 
\end{eqnarray}
which tends to the limit point $\rho_\infty = \textup{diag} ( (\lambda_1+\lambda_2)/2 \, , \,  (\lambda_1+\lambda_2)/2 )$. \\

As a second case, we take $n=3$, i.e. the case of a qutrit, whose diagonal density matrices are of the form: $\rho= \textup{diag}(\lambda_1,\lambda_2, \lambda_3=1-\lambda_1-\lambda_2)$.  Setting $X_1= (\lambda_1-\lambda_2)/2$ and $X_2= (\lambda_1+\lambda_2)/2-1/3$, we can represent the simplex $\Delta^2$ in the $X_1-X_2$-plane as the (equilateral) triangle with vertices: $P_1=(1,\sqrt{3}), \, P_2=(-1,\sqrt{3}), \, P_3=(0,-2/\sqrt{3})$, as shown in Fig. 1(b). \\
Now there are different kinds of cyclic subgroups.\\
For example, we can assume $S_1 : (\lambda_1,\lambda_2,\lambda_3)\mapsto (\lambda_1, \lambda_3,\lambda_2)$, which corresponds to a cycle of length 1 and one of length 2. Then, the density matrix evolves in time through a UDM $F_1$, as follows:
\begin{equation}
\rho_{F_1}(t) =\left( \begin{array}{ccc} \lambda_1&~ &~\\
~& e^{-t}  \lambda_2 + (1-e^{-t} ) (\lambda_2+\lambda_3)/2 & ~ \\
~&~& e^{-t}  \lambda_3 + (1-e^{-t} ) (\lambda_2+\lambda_3)/2  \end{array} \right)
\end{equation}
which tends to the limit point $\rho_\infty = \textup{diag} (\lambda_1\, , \,  (\lambda_2+\lambda_3)/2\, , \,  (\lambda_2+\lambda_3)/2 )$. As it is shown in Fig. 1(b), the orbit is parallel to the side $P_2P_3$ of the triangle. Similar orbits, but now parallel to the other sides $P_1P_3$ and $P_1P_2$ are obtained by considering the cyclic subgroups: 
$S_2 : (\lambda_1,\lambda_2,\lambda_3)\mapsto (\lambda_3, \lambda_2,\lambda_1)$ and $S_3 : (\lambda_1,\lambda_2,\lambda_3)\mapsto (\lambda_2, \lambda_1,\lambda_3)$ respectively.\\
We can also consider the maximal cyclic subgroup $S : (\lambda_1,\lambda_2,\lambda_3)\mapsto (\lambda_3, \lambda_1,\lambda_2)$, which yields:
\begin{equation}
\rho_{F}(t)  =\left( \begin{array}{ccc} e^{-t}  \lambda_1 +  (1-e^{-t} )/3 &~ & ~\\
~& e^{-t}  \lambda_2 +  (1-e^{-t} ) /3  & ~ \\
~&~& e^{-t}  \lambda_3 +  (1-e^{-t} ) /3   \end{array} \right) 
\end{equation}
whose limit point is the barycenter of the triangle, i.e. the maximally mixed matrix  $\rho_\infty= \textup{diag}(1/3,1/3,1/3)$.

\bigskip

\begin{flushleft}

\vbox{
Alessia~CATTABRIGA\\
Department of Mathematics, University of Bologna\\
Piazza di Porta San Donato 5, 40126 Bologna, ITALY\\
e-mail: \texttt{alessia.cattabriga@unibo.it}}

\bigskip

\vbox{
Elisa~ERCOLESSI\\
Department of  Physics and Astronomy, University of Bologna\\
and INFN, Sezione di Bologna, via Irnerio 46, 40126 Bologna, ITALY\\
e-mail: \texttt{elisa.ercolessi@unibo.it}}

\bigskip

\vbox{
Riccardo~GOZZI\\
Instituto Superior Tecnico, Universidade de Lisboa\\
 Av. Rovisco Pais 1, 1049-001\\
  Instituto de Telecomunicaçoe, Lisboa, PORTUGAL\\
e-mail: \texttt{ilgozzi@mail.com}}

\bigskip

\vbox{
Erika~MEUCCI\\
School of Advanced International Studies, Johns Hopkins University\\
 via Beniamino Andreatta 3, 40126 Bologna, ITALY\\
e-mail: \texttt{emeucci@jhu.edu}}

\end{flushleft}


\begin{thebibliography}{0}

\bibitem{EMM} E. Ercolessi, G. Marmo, and G. Morandi, {\it From the equations of motion to the canonical commutation relations}, La Rivista del Nuovo Cimento 33 (2010) 401.

\bibitem{PRE} J. Preskill, {\it Lecture Notes for Physics 229:Quantum Information and Computation}, CreateSpace Independent Publishing Platform (2015).

\bibitem{CCM} J. F. Carin\~{e}na, J. Clemente-Gallardo and G. Marmo, {\it Geometrization of quantum mechanics}, Theoretical and Mathematical Physics  152(1) (2007) 894. 

\bibitem{CES} I. Contreras, E. Ercolessi and M. Schiavina, {\it On the geometry of mixed states and the Fisher information tensor}, Jour. Math. Phys. 57 (2016) 062209.

\bibitem{CP} D. Chru\'sci\'nski and S.  Pascazio, {\it A Brief History of the GKLS Equation},  Open Sys. Inf. Dyn. 24 (2017) 1740001. 

\bibitem{CCILM} F.M. Ciaglia, F. Di Cosmo, A. Ibort, M. Laudato and G. Marmo, Dynamical Vector Fields on the Manifold of Quantum States, Open Systems and Information Dynamics, 24 (2017) 1740003.

\bibitem{RH} A. Rivas and S.F. Huelga, {\it Open Quantum Systems. An Introduction}, Springer, Heidelberg- Dordrecht-London-New York (2012).

\bibitem{S} B.E. Sagan, {\it The symmetric group: representations, combinatorial algorithms, and symmetric functions},  Springer-Verlag, New York, (2001).

\bibitem{Choi} M.D. Choi, {\it Completely Positive Linear Maps on Complex Matrices}, Linear Algebra and its Applications 10 (1975).



\end{thebibliography}
\end{document}